\documentclass[a4paper,UKenglish]{llncs}

\usepackage[T1]{fontenc}

\usepackage[utf8]{inputenc}

\usepackage{microtype}
\usepackage{lmodern}

\usepackage{pgfplots}

\usepackage{amssymb,amsmath}

\usepackage{xspace,relsize,stmaryrd}
\usepackage{amssymb,amsmath}
\usepackage{mathpartir}
\usepackage{bussproofs}
\usepackage{multirow}

\usepackage[author=anonymous,margin,footnote,draft]{fixme}
\FXRegisterAuthor{dh}{adh}{DH}
\FXRegisterAuthor{ls}{als}{LS}

\allowdisplaybreaks[2]

\usepackage{tikz}
 \usetikzlibrary{trees}
 \usetikzlibrary{shapes}
 \usetikzlibrary{fit}
 \usetikzlibrary{shadows}
 \usetikzlibrary{backgrounds}
 \usetikzlibrary{arrows,automata}

\tikzset{
   n/.style= {circle,fill,inner sep=1.5pt,node distance=2cm}
  ,acc/.style={circle,draw,inner sep=3pt,node distance=2cm}
  ,phantom/.style={circle},
  ,arr/.style={->, >=stealth, semithick, shorten <= 3pt, shorten >= 3pt}
}

\newcounter{blubber}

\makeatletter
\def\moverlay{\mathpalette\mov@rlay}
\def\mov@rlay#1#2{\leavevmode\vtop{%
   \baselineskip\z@skip \lineskiplimit-\maxdimen
   \ialign{\hfil$\m@th#1##$\hfil\cr#2\crcr}}}
\newcommand{\charfusion}[3][\mathord]{
    #1{\ifx#1\mathop\vphantom{#2}\fi
        \mathpalette\mov@rlay{#2\cr#3}
      }
    \ifx#1\mathop\expandafter\displaylimits\fi}
\makeatother

\newcommand\pfun{\mathrel{\ooalign{\hfil$\mapstochar\mkern5mu$\hfil\cr$\to$\cr}}}

\newcommand{\detcarrier}{D_\target}
\newcommand{\detprio}{\beta}
\newcommand{\FL}{\mathsf{FL}}

\newcommand{\target}{\chi}
\newcommand{\FLtarget}{\mathbf F}

\newcommand\ExpTime{$\textsc{ExpTime}$\xspace}

\newcommand{\hearts}{\heartsuit}

\newcommand{\sem}[1]{[\![#1]\!]}

\newcommand{\psem}[1]{\widehat{[\![#1]\!]}}

\newcommand{\Pow}{\mathcal{P}}

\spnewtheorem{thm}[theorem]{Theorem}{\bfseries}{\itshape}
\spnewtheorem{cor}[theorem]{Corollary}{\bfseries}{\itshape}
\spnewtheorem{cnj}[theorem]{Conjecture}{\bfseries}{\itshape}
\spnewtheorem{lem}[theorem]{Lemma}{\bfseries}{\itshape}
\spnewtheorem{lemdefn}[theorem]{Lemma and Definition}{\bfseries}{\itshape}
\spnewtheorem{prop}[theorem]{Proposition}{\bfseries}{\itshape}
\spnewtheorem{defn}[theorem]{Definition}{\bfseries}{\upshape}
\spnewtheorem{rem}[theorem]{Remark}{\bfseries}{\upshape}
\spnewtheorem{notation}[theorem]{Notation}{\bfseries}{\upshape}
\spnewtheorem{expl}[theorem]{Example}{\bfseries}{\upshape}
\spnewtheorem{thmdefn}[theorem]{Theorem and Definition}{\bfseries}{\itshape}
\spnewtheorem{propdefn}[theorem]{Proposition and Definition}{\bfseries}{\itshape}
\spnewtheorem{assumption}[theorem]{Assumption}{\bfseries}{\upshape}
\spnewtheorem{algorithm}[theorem]{Algorithm}{\bfseries}{\upshape}

 \renewenvironment{theorem}{\begin{thm}}{\end{thm}}
 \renewenvironment{corollary}{\begin{cor}}{\end{cor}}
 \renewenvironment{lemma}{\begin{lem}}{\end{lem}}
 
 \renewenvironment{definition}{\begin{defn}}{\end{defn}}
 \renewenvironment{remark}{\begin{rem}}{\end{rem}}
 \renewenvironment{example}{\begin{expl}}{\end{expl}}

\begin{document}

\bibliographystyle{myabbrv}

\title{Optimal Satisfiability Checking\\ for Arithmetic $\mu$-Calculi\protect}

\author{Daniel Hausmann \and Lutz Schr\"oder}
\institute{Friedrich-Alexander-Universit\"{a}t
Erlangen-N\"urnberg, Germany}

\maketitle

\begin{abstract} The coalgebraic $\mu$-calculus provides a generic semantic framework for fixpoint logics with branching types beyond the standard relational setup, e.g.\ probabilistic, weighted, or game-based. Previous work on the coalgebraic $\mu$-calculus includes an exponential time upper bound on satisfiability checking, which however requires a well-behaved set of tableau rules for the next-step modalities. Such rules are not available in all cases of interest, in particular ones involving either integer weights as in the graded $\mu$-calculus, or real-valued weights in combination with non-linear arithmetic. In the present paper, we prove the same upper complexity bound under more general assumptions, specifically regarding the complexity of the (much simpler) satisfiability problem for the underlying so-called one-step logic, roughly described as the nesting-free next-step fragment of the logic. We also present a generic global caching algorithm that is suitable for practical use and supports on-the-fly satisfiability checking. Example applications include new exponential-time upper bounds for satisfiability checking in an extension of the graded $\mu$-calculus with Presburger arithmetic, as well as an extension of the (two-valued) probabilistic $\mu$-calculus with polynomial inequalities. As a side result, we moreover obtain a new upper bound $\mathcal{O}(((nk)!)^2)$ on minimum model size for satisfiable formulas for \emph{all} coalgebraic $\mu$-calculi, where~$n$ is the size of the formula and~$k$ its alternation depth.

\end{abstract}

\section{Introduction}

Modal fixpoint logics are a well-established tool in the temporal
specification, verification, and analysis of concurrent systems. One
of the most expressive logics of this type is the modal
$\mu$-calculus~\cite{Kozen88,BradfieldStirling06,BradfieldWalukiewicz18},
which features explicit operators for taking least and greatest
fixpoints, which roughly speaking serve the specification of liveness
properties (least fixpoints) and safety properties (greatest
fixpoints), respectively. Like most modal logics, the modal
$\mu$-calculus is traditionally interpreted over relational models
such as Kripke frames or labelled transition systems. The growing
interest in more expressive models where transitions are governed,
e.g., by probabilities, weights, or games has sparked a commensurate
growth of temporal logics and fixpoint logics interpreted over such
systems; prominent examples include probabilistic
$\mu$-calculi~\cite{CleavelandEA05,HuthKwiatkowska97,LiuEA15}, the
alternating-time $\mu$-calculus~\cite{AlurEA02}, and the monotone
$\mu$-calculus, which contains Parikh's game
logic~\cite{Parikh85}. The graded $\mu$-calculus~\cite{KupfermanEA02}
features next-step modalities that count successors; it is standardly
interpreted over Kripke frames but, as pointed out by D'Agostino and
Visser~\cite{DAgostinoVisser02}, graded modalities are more naturally
interpreted over so-called multigraphs, where edges carry integer
weights, and in fact we shall see that this modification leads to
better bounds on minimum model size for satisfiable formulas.

Coalgebraic logic~\cite{Pattinson03,Schroder08} has emerged as a
unifying framework for modal logics interpreted over such more general
models; it is based on the principle of casting the transition type of
the systems at hand as a set functor, and the systems in question as
coalgebras for this type functor, following the paradigm of universal
coalgebra~\cite{Rutten00}; additionally, modalities are interpreted as
so-called \emph{predicate liftings}. The \emph{coalgebraic
  $\mu$-calculus}~\cite{CirsteaEA11a} caters for fixpoint logics
within this framework, and essentially covers all mentioned
(two-valued) examples as instances. It has been shown that
satisfiability checking in a coalgebraic $\mu$-calculus is in
\ExpTime, \emph{provided} that one exhibits a set of tableau rules for
the modalities, so-called \emph{one-step rules}, that is
\emph{tractable} in a suitable sense. Such rules are known for many
important cases, notably including alternating-time logics, the
probabilistic $\mu$-calculus even when extended with linear
inequalities, and game
logic~\cite{SchroderPattinson09,KupkePattinson10,CirsteaEA11a}. There
are, however, important cases where such rule sets are currently
missing, and where there is in fact little perspective for finding
suitable rules. One prominent case of this kind is graded modal logic;
further cases arise when logics over systems with non-negative real
weights, such as probabilistic systems, are taken beyond linear
arithmetic to include polynomial inequalities.

The object of the current paper is to fill this gap by proving a
generic upper bound \ExpTime for coalgebraic $\mu$-calculi in the
absence of tractable sets of modal tableau rules. The method we use
instead is to analyse the so-called \emph{one-step satisfiability}
problem of the logic on a semantic level -- this problem is
essentially the satisfiability problem of a very small fragment of the
logic, the \emph{one-step logic}, which excludes not only fixpoints,
but also nested next-step modalities, with a correspondingly
simplified semantics that no longer involves actual transitions. E.g.\
the one-step logic of the relational $\mu$-calculus is interpreted
over models essentially consisting of a set with a distinguished
subset, abstracting the successors of a single state that is not
itself part of the model. We have applied this principle to
satisfiability checking in coalgebraic (next-step) modal
logics~\cite{SchroderPattinson08}, coalgebraic hybrid
logics~\cite{MyersEA09}, and reasoning with global assumptions in
coalgebraic modal logics~\cite{KupkeEA15}. It also appears implicitly
in work on automata for the coalgebraic
$\mu$-calculus~\cite{FontaineEA10}, which however establishes only a
doubly exponential upper bound in the case without tractable modal
tableau rules.

Our main example applications are on the one hand the graded modal
$\mu$-calculus and its extension with Presburger modalities, i.e.\
with (monotone) linear inequalities, and on the other hand the
extension of the (two-valued) probabilistic
$\mu$-calculus~\cite{CirsteaEA11a,LiuEA15} with (monotone) polynomial
inequalities. While the graded $\mu$-calculus as such is known to be
in \ExpTime~\cite{KupfermanEA02}, the other mentioned instances of our
result are, to our best knowledge, new. At the same time, our proofs
are fairly simple, even compared to specific ones, e.g.\ for the
graded $\mu$-calculus.

Technically, we base our results on an automata-theoretic treatment
by means of standard parity automata with singly-exponential branching degree
(in particular on modal steps), thus precisely enabling the
singly-exponential upper bound, in contrast to previous work in~\cite{FontaineEA10} where the introduced $\Lambda$-automata lead to doubly-exponential branching on modal steps in the resulting satisfiability games.
Our new algorithm for satisfiability
witnessing the singly-exponential time bound is, in fact, a global caching
algorithm~\cite{GoreWidmann09,GoreNguyen13}, and is able to decide the
satisfiability of nodes on-the-fly, that is, possibly 
before the tableau is fully expanded, thus offering a
perspective for practically feasible reasoning. A side result of our
approach is a singly-exponential bound on minimum model size for
satisfiable formulas for \emph{all}
coalgebraic $\mu$-calculi, calculated only in terms of the size of the
parse tree of the formula and its alternation depth (again, the best
previously known bound for the case without tractable modal tableau
rules~\cite{FontaineEA10} was doubly exponential). This bound is new
even in the case of the graded $\mu$-calculus \emph{over multigraphs}
-- over Kripke frames, it is clearly just not true as the model size
can depend exponentially on numbers occurring in a formula when these
are coded in binary, again illustrating the smoothness of multigraph
semantics. Moreover, we identify a criterion for a polynomial bound on
branching in models, which holds in all our examples.

The material is organized as follows. In 
Section~\ref{sec:prelims}, we recall the basics of
coalgebra and the coalgebraic $\mu$-calculus. We
outline our automata-theoretic approach in Section~\ref{sec:tracking},
and present the global caching algorithm and its runtime analysis in
Section~\ref{section:alg}. Soundness and completeness of the algorithm
are proved in Section~\ref{sec:correctness}.


\section{The Coalgebraic $\mu$-Calculus}\label{sec:prelims}

We recall the basics of the framework of coalgebraic
logic~\cite{Pattinson03,Schroder08} and the coalgebraic
$\mu$-calculus~\cite{CirsteaEA11a}. For ease of notation, we restrict
the technical development to unary modalities in this work, noting
that all proofs naturally generalize to the $n$-ary setting; in fact,
we will liberally use higher arities in examples.  We fix a
$\mathbf{Set}$-endofunctor $T$, where elements of~$TX$ should be
regarded as structured collections over~$X$ that will later serve as
collections of successors of states (in the most basic example, $T$ is
powerset~$\Pow$), and a \emph{modal similarity type} $\Lambda$, that is, a
set of unary modal operators. We assume that $\Lambda$ is closed under
duals, i.e., that for each modal operator $\hearts\in\Lambda$, there
is a \emph{dual} $\overline{\hearts}\in\Lambda$ such that
$\overline{\overline{\hearts}}=\hearts$ for all $\hearts\in\Lambda$.
We interpret modal operators $\hearts\in\Lambda$ as
\emph{$T$-predicate liftings}, that is, natural transformations
$\sem{\hearts}:\mathcal{Q}\to \mathcal{Q}\circ T^{\mathit{op}}$ where
$\mathcal{Q}:\mathbf{Set}^{\mathit{op}}\to\mathbf{Set}$
denotes the contravariant powerset functor.  Predicate liftings thus
are just families of functions
$\sem{\hearts}_X:\mathcal{Q}(X)\to\mathcal{Q}(TX)$ that satisfy
\emph{naturality}, i.e.\
$\sem{\hearts}_X(f^{-1}[A])= (T f)^{-1}[\sem{\hearts}_Y(A)]$ for all
$X,Y\in\mathbf{Set}$, all $f:X\to Y$ and all $A\subseteq Y$, where
$f^{-1}$ denotes preimage.  We assume that $\Lambda$ comes with a
predicate lifting $\sem{\hearts}$ for each $\hearts\in\Lambda$;
furthermore we require that the duality of modal operators is
respected, i.e. that
$\sem{\hearts}_V(U)=\overline{\sem{\overline{\hearts}}_V\overline{U}}$
for all sets $V$, $U\subseteq V$, where for all sets $U$ and the
according obvious base set $V$,
$\overline{U}=\{u\in V\mid u\notin U\}$ denotes the \emph{complement}
of $U$ in $V$.  Given a set $U$, a function $f:\Pow(U)\to\Pow(U)$ is
\emph{monotone} if for all $A,B\subseteq U$, $A\subseteq B$ implies
$f(A)\subseteq f(B)$.  To ensure the existence of fixpoints of
formulas, we require that all predicate liftings are monotone.

\begin{definition}[Coalgebraic $\mu$-calculus~\cite{CirsteaEA11a}]
Let $\mathbf{V}$ be an infinite set of \emph{fixpoint variables}.
Formulas of the \emph{coalgebraic $\mu$-calculus} (over $\Lambda$)
are given by the grammar
\begin{align*}
\psi,\phi ::= \bot \mid \top \mid \psi\wedge\phi \mid 
\psi\vee\phi \mid \hearts\phi \mid X \mid \mu X.\,\psi \mid \nu X.\,\psi\qquad
\qquad \hearts\in \Lambda, X\in\mathbf{V}
\end{align*}
Formulas are interpreted over \emph{$T$-coalgebras}, that is, pairs
$(C,\xi)$, consisting of a set $C$ of \emph{states} and a
\emph{transition function} $\xi:C\to TC$ that assigns a structured
collection $\xi(x)\in TC$ of successors (and observations) to
$x\in C$; e.g.\ $\Pow$-coalgebras are just Kripke frames. The
valuation of fixpoint variables requires partial functions
$i:\mathbf{V}\pfun \Pow(C)$ that assign sets $i(X)$ of states to
fixpoint variables $X$. To interpret formulas over $(C,\xi)$, we
define the expected clauses for propositional formulas plus
\begin{align*}
\sem{\hearts\psi}_i &= \xi^{-1}[\sem{\hearts}_C(\sem{\psi}_i)]&
\sem{\mu X.\,\psi}_i &= \mathsf{LFP}(\sem{\psi}^X_i)\\
\sem{X}_i &= i(X)&
\sem{\nu X.\,\psi}_i &= \mathsf{GFP}(\sem{\psi}^X_i),
\end{align*}
where $\mathsf{LFP}$ and $\mathsf{GFP}$ compute the least and greatest
fixpoints of their argument functions, respectively, where
$\sem{\psi}^X_i(A)=\sem{\psi}_{i[X\mapsto A]}$ for $A\subseteq C$ and
where $(i[X\mapsto A])(X)=A$ and $(i[X\mapsto A])(Y)=i(Y)$ for
$Y\neq X$.  Thus we have $x\in\sem{\hearts\psi}_i$ if and only if
$\xi(x)\in\sem{\hearts}_C(\sem{\psi}_i)$.  By the monotonicity of
predicate liftings, the extremal fixpoints of the functions
$\sem{\psi}^X_i$ are indeed defined.  Although the logic does not
contain negation as an explicit operator, negation can be defined by
taking negation normal forms. Similarly, the framework does not force
the inclusion of propositional atoms, which can just be seen as
nullary modalities (see Example~\ref{ex:logics}.1.).  Fixpoint
operators \emph{bind} their fixpoint variables, so that we have
standard notions of bound and free fixpoint variables; a formula is
closed if it contains no free fixpoint variables. For closed formulas
$\psi$, the valuation of fixpoint variables is irrelevant so that we
write $\sem{\psi}$ instead of $\sem{\psi}_i$. A state $x\in C$
\emph{satisfies} a closed formula $\psi$ (denoted $x\models \psi$) if
$x\in\sem{\psi}$.  Given a set~$V$, we put
$\Lambda(V)=\{\hearts a\mid \hearts\in\Lambda, a\in V\}$ and refer to
elements $\hearts a\in \Lambda(V)$ as \emph{modal literals} (over
$V$).  Throughout, we use $\eta\in\{\mu,\nu\}$ to denote extremal
fixpoint operators.  The \emph{size} $|\psi|$ of a formula is its
length over the alphabet
$\{\bot,\top,\wedge,\vee\}\cup \Lambda\cup \mathbf{V} \cup \{\eta
X.\mid X\in\mathbf{V} \}$,
where we assume that the length of $\hearts\in\Lambda$
is the size of its representation.
The \emph{alternation depth} $\mathsf{ad}(\psi)$ of a formula $\psi$
is the depth of dependent nesting of alternating least and greatest
fixpoints in $\psi$; we assign \emph{even} numbers to least fixpoint
formulas and \emph{odd} numbers to greatest fixpoint formulas, and, as
usual, assign greater numbers to outermost fixpoints. For a more
detailed definition of various flavours of alternation depth, see
e.g.~\cite{NiwinskiWalukiewicz96}. The \emph{satisfiability problem}
of the coalgebraic $\mu$-calculus is to decide, for a given formula
$\target$, whether there is a coalgebra $(C,\xi)$ and a state $x\in C$
that satisfies $\target$.  We restrict our development to formulas in
which all fixpoint variables are guarded by modal operators;
furthermore, we assume w.l.o.g. that all formulas are \emph{clean},
i.e\ that each fixpoint variable is bound by at most one fixpoint
operator, and \emph{irredundant}, i.e\ each bound variable is used at
least once.
\end{definition}
\noindent As usual in $\mu$-calculi, the unfolding of fixpoints does
not affect their semantics, that is, for all $X\in\mathbf{V}$ and all
formulas $\psi$, we have
$\sem{\eta X.\,\psi}=\sem{\psi[X\mapsto \eta X.\,\psi]}$.

\begin{example}\label{ex:logics}
We now detail several instances of the coalgebraic $\mu$-calculus; for
further examples, e.g. the alternating-time $\mu$-calculus, see~\cite{CirsteaEA11a}.
\begin{enumerate}
\item To obtain the standard modal $\mu$-calculus~\cite{Kozen83}
  (which contains CTL as a simple fragment), we use the powerset
  functor, that is, we put $T=\mathcal{P}$ so that $T$-coalgebras are
  Kripke frames. To ease readability, we refrain from incorporating
  propositional atoms into the logic, noting that atoms from a set $P$
  can easily be added by switching to the functor $T\times \Pow(P)$
  and then defining their semantics by means of nullary predicate
  liftings.  As modal similarity type, we take
  $\Lambda=\{\Diamond, \Box\}$ and define the predicate liftings
\begin{align*}
\sem{\Diamond}_U(A)&=\{B\in\Pow (U)\mid A\cap B\neq \emptyset\}& 
\sem{\Box}_U(A)&=\{B\in\Pow (U)\mid B\subseteq A\}
\end{align*}
for sets $U$ and $A\subseteq U$.  Standard examples include the
CTL-formula $\mathsf{AF}\,\psi=\mu X.\,(\psi\vee \Box X)$, which
states that on all paths, $\psi$ eventually holds, and the fairness
formula $\nu X.\,\mu Y.\, ((\psi\wedge \Diamond X)\vee \Diamond Y)$,
which asserts the existence of a path on which the formula $\psi$ is
satisfied infinitely often.
\item We interpret the \emph{graded
    $\mu$-calculus}~\cite{KupfermanEA02} over
  multigraphs~\cite{DAgostinoVisser02}, i.e.\ $T$-coalgebras for the
  multiset functor $T=\mathcal{B}$, defined by
\begin{align*}
\mathcal{B}(U)&=\{\theta:U\to \mathbb{N}\cup\{\infty\}\} &
\mathcal{B}(f)(\theta)(v)&=\sum\nolimits_{u\in U\mid f(u)=v}\theta(u)
\end{align*}
for sets $U,V$ and functions $f:U\to V$,
$\theta:U\to \mathbb{N}\cup\{\infty\}$.  Thus $\mathcal{B}$-coalgebras
$(C,\xi)$ assign multisets $\xi(x)$ to states $x\in C$, with the
intuition that $x$ has $y\in C$ as successor with multiplicity $k$ if
$(\xi(x))(y)=k$.  We use the modal similarity type
$\Lambda=\{\langle k\rangle,[k]\mid k\in\mathbb{N}\cup\{\infty\}\}$
and define the predicate liftings
\begin{align*}
\sem{\langle k\rangle}_U(A)&=\{\theta\in\mathcal{B}(X)\mid \theta(A)>k\} &
\sem{[k]}_U(A)&=\{\theta\in\mathcal{B}(X)\mid \theta(\overline{A})\leq k\}
\end{align*}
for sets $U$ and $A\subseteq U$, where $\theta(A)=\sum_{a\in A} \theta(a)$.
E.g. the formula $\nu X.\,(\psi\wedge\langle 1 \rangle X)$
expresses the existence of an infinite binary tree in which
the formula $\psi$ is satisfied
globally.
\item Similarly, the two-valued 
\emph{probabilistic $\mu$-calculus}~\cite{CirsteaEA11a,LiuEA15} is obtained
by using the distribution functor $T=\mathcal{D}$
that maps sets $U$ to probability distributions over $U$ with countable support, defined by
\begin{align*}
\mathcal{D}(U)=\{d:U\to (\mathbb{Q}\cap[0,1])\mid \textstyle\sum\nolimits_{u\in U} d(u)=1\}.
\end{align*}
Then $\mathcal{D}$-coalgebras are just Markov chains.
We use the modal similarity type $\Lambda=\{\langle p\rangle,[p]\mid p\in\mathbb{Q}\cap[0,1])\}$
and define the predicate liftings
\begin{align*}
\sem{\langle p\rangle}_U(A)&=\{d\in\mathcal{D}(X)\mid d(A)>p\}& 
\sem{[p]}_U(A)&=\{d\in\mathcal{D}(X)\mid d(\overline{A})\leq p\},
\end{align*}
for sets $U$ and $A\subseteq U$, where again $d(A)=\sum_{a\in A} d(a)$.
\item The \emph{Presburger $\mu$-calculus} is the extension of the
  graded $\mu$-calculus with Presburger arithmetic; the next step
  version of the logic was introduced by Demri and
  Lugiez~\cite{DemriLugiez06}. Its formulas can be interpreted over
  the semantic domain from item 2., that is, over multigraphs.  We
  introduce new higher-arity modalities by putting
  $\Lambda=\{L_{a_1,\ldots,a_n, b},M_{a_1,\ldots,a_n,b}\mid
  a_1,\ldots,a_n,b,n\in\mathbb{N}\}$
  and define the predicate liftings
\vspace{-5pt}
\begin{align*}
\sem{L_{a_1,\ldots,a_n, b}}_U(A_1,\ldots,A_n)&=\{\theta\in\mathcal{B}(X)\mid \textstyle\sum\nolimits^n_{i=1}a_i\cdot \theta(A_i)> b\}\\
\sem{M_{a_1,\ldots,a_n, b}}_U(A_1,\ldots,A_n)&=\{\theta\in\mathcal{B}(X)\mid \textstyle\sum\nolimits^n_{i=1}a_i\cdot \theta(\overline{A_i})\leq b\},
\end{align*}
for sets $U$ and $A_1,\ldots,A_n\subseteq U$, where $\theta(A)=\textstyle\sum_{a\in A} \theta(a)$.

\item Similarly, we use the semantic domain from item 3., Markov
  chains, to obtain the \emph{probabilistic $\mu$-calculus with
    polynomial inequalities}~\cite{KupkeEA15}.  Again, we introduce
  new higher-arity modalities by putting
  $\Lambda=\{L_{p,b},M_{p,b}\mid p\in \mathbb{Q}_{>0}[X_1,\ldots, X_n],b\in\mathbb{Q}_{\ge 0}\}$
  (i.e.~$p$ ranges over polynomials) and 
\vspace{-5pt}
\begin{align*}
\sem{L_{p,b}}_U(A_1,\ldots,A_n)&=\{d\in\mathcal{D}(X)\mid 
p(d(A_1),\ldots,d(A_n))> b\}\\
\sem{M_{p,b}}_U(A_1,\ldots,A_n)&=\{d\in\mathcal{D}(X)\mid 
p(d(\overline{A_1}),\ldots,d(\overline{A_n}))\leq b\}
\end{align*}
for sets $U$ and $A_1,\ldots,A_n\subseteq U$, where
again $f(A)=\sum_{a\in A} f(a)$.

\end{enumerate}
The logics from the last two items are necessarily less general than
the corresponding next-step
logics~\cite{DemriLugiez06,KupkeEA15}, because the definition of
$\mu$-calculi requires monotonicity of the involved predicate
liftings. To ensure monotonicity, we restrict all coefficients
to be positive, and moreover we restrict the relation in item 4. 
to be $>$ instead of one of the relations
$\{>,<,=\}\cup\{\equiv_k\mid k\in\mathbb{N}\}$.
\end{example}

\section{Tracking Automata}\label{sec:tracking}

We use \emph{parity automata} (e.g.~\cite{Graedel02}) that track
single formulas along paths through potential models to decide whether
it is possible to construct a model in which all least fixpoint
formulas are eventually satisfied. Formally, (nondeterministic) parity
automata are tuples $\mathsf{A}=(V,\Sigma,\Delta,q_0,\alpha)$ where
$V$ is a set of \emph{nodes}; $\Sigma$ is a finite set, the
\emph{alphabet}; $\Delta\subseteq V\times \Sigma\times V$ is the
\emph{transition relation} assigning a set
$\Delta(v,a)=\{u\mid (v,a,u)\in\Delta\}$ of nodes to all $v\in V$ and
$a\in \Sigma$; $q_0\in V$ is the \emph{initial node}; and
$\alpha:V\to\mathbb{N}$ is the \emph{priority function}, assigning
priorities $\alpha(v,a,u)\in\mathbb{N}$ to \emph{transitions} $(v,a,u)\in \Delta$ (this is the standard in recent work since it yields slightly more
succinct automata).  If $\Delta$ is
a function, then $\mathsf{A}$ is said to be \emph{deterministic}.  
The
automaton $\mathsf{A}$ \emph{accepts} an infinite word
$w=w_0w_1,\ldots\in\Sigma^\omega$ if there is a $w$-path through
$\mathsf{A}$ on which the highest priority that is passed infinitely
often is even; formally, the language that is accepted by $\mathsf{A}$
is defined by
$L(\mathsf{A})=\{w\in\Sigma^\omega\mid \exists
\rho\in\mathsf{run}(\mathsf{A},w).\,
\max(\mathsf{Inf}(\alpha\circ\rho))\text{ is even}\}$,
where $\mathsf{run}(\mathsf{A},w)$ denotes the set of infinite
sequences $\rho_0,\rho_1,\ldots\in V^\omega$ such that $\rho_0=q_0$
and for all $i\geq 0$,
$\rho_{i+1}\in \Delta(\rho_i,w_i)$ and where, given an infinite sequence
$S$, $\mathsf{Inf}(S)$ denotes the elements that occur infinitely
often in $S$. 
Here, we see infinite sequences $\rho\in U^\omega$ over
some set $U$ as functions $\mathbb{N}\to U$ and write $\rho_i$ to
denote the $i$-th element of~$\rho$.

We now fix a target formula $\target$ and put $n:=|\target|$
and $k:=\mathsf{ad}(\target)$.
We let $\FLtarget:=\FL(\target)$ denote the \emph{Fischer-Ladner closure}~\cite{Kozen88}
of $\target$; the Fischer-Ladner closure contains all formulas that can arise as subformulas when unfolding each fixpoint in $\target$ exactly once. We put $\mathsf{selections}:=\Pow(\FLtarget\cap\Lambda(\FLtarget))$
where $\FLtarget\cap\Lambda(\FLtarget)$ is the set of formulas
from $\FLtarget$ that are modal literals.
We have $|\FLtarget|\leq n$ and hence $|\mathsf{selections}|\leq 2^n$.

\begin{definition}[Tracking automaton]
The \emph{tracking automaton} for a formula~$\target$ is a nondeterministic parity
automaton $\mathsf{A}_\target=(\FLtarget,\Sigma,\Delta,q_0,\alpha)$,
where $q_0=\target$,
\vspace{-3pt}
\begin{align*}
\Sigma=&
\{(\psi_0\vee\psi_1,b)\in\FLtarget\times \{0,1\}\}\cup
\{(\psi_0\wedge\psi_1,0)\in \FLtarget\times\{0\}\}\cup\\
&\{(\eta X.\,\psi_1,0)\in\FLtarget\times\{0\}\}
\cup\mathsf{~selections~},
\end{align*}
for $\psi,\psi_0,\psi_1\in\FLtarget$, $\sigma\in\mathsf{selections}$
and $b\in \{0,1\}$,
\vspace{-3pt}
\begin{align*}
\Delta(\psi,\sigma)&=\{\psi_0\in\FLtarget\mid \psi\in\sigma\cap\Lambda(\{\psi_0\})\},\\
\Delta(\psi,(\psi_0\vee\psi_1,b))&=\{\psi_b\mid \psi=\psi_0\vee\psi_1 \}\cup\{\psi\mid\psi\neq\psi_0\vee\psi_1\}\\
\Delta(\psi,(\psi_0\wedge\psi_1,0))&=\{\psi_0,\psi_1\mid \psi=\psi_0\wedge\psi_1 \}\cup\{\psi\mid\psi\neq\psi_0\wedge\psi_1\}\\
\Delta(\psi,(\eta X.\,\psi_1,0))&=\{\psi_1[X\mapsto\psi]\mid \psi=\eta X.\,\psi_1 \}\cup\{\psi\mid\psi\neq\eta X.\,\psi_1\}.
\end{align*}
E.g.\ the last clause means that when tracking the unfolding of a
fixpoint $\eta X.\,\psi_1$ at $\psi$, we track~$\psi$ to the unfolding
$\psi_1[X\mapsto\psi]$ if $\psi$ equals the unfolded fixpoint, and
to~$\psi$ otherwise; similarly for the other clauses, and in particular
a modal literal $\psi=\hearts\psi_0$ is only tracked to $\psi_0$ through
a selection $\sigma$ if $\hearts\psi_0\in\sigma$, i.e. if
$\sigma$ selects $\hearts\psi_0$ to be tracked. The priority
function $\alpha$ is derived from the alternation depths of formulas,
counting only unfoldings of fixpoints (i.e. all other transitions have
priority 1). Formally, $\alpha(\psi,\sigma,\psi')=1$ if $\psi=\psi'$
or $\psi$ is not a fixpoint literal; if $\psi$ is a fixpoint literal
and $\psi\neq \psi'$, then we put
$\alpha(\psi,\sigma,\psi')=\mathsf{ad}(\psi)$.
\end{definition}
\noindent Intuitively, words from $\Sigma^\omega$ encode infinite paths
through labelled coalgebras $(C,\xi)$ where letters $\sigma\in 
\mathsf{selections}$
encode modal steps from states $x\in C$ with label $l(x)$ to states
$y\in C$ with label
$\{\psi\in \FLtarget\mid\exists \hearts\in\Lambda.\,
\hearts\psi\in\sigma\cap l(x)\}$.
Letters $(\psi_0\vee\psi_1,b)$ choose disjuncts; the tracking
automaton is nondeterministic for letters $(\psi_0\wedge\psi_1,0)$ and
accepts exactly the words that encode a path that contains a least
fixpoint formula~$\psi$ that is unfolded infinitely often without
being dominated by any outer fixpoint formula (i.e. one with
alternation depth greater than $\mathsf{ad}(\psi)$); denoting these
words by $\mathsf{BadBranch}_\target$, we thus have
$L(\mathsf{A}_\target)=\mathsf{BadBranch}_\target$.  The automaton
$\mathsf{A}_\target$ has size $n$ and priorities $1$ to
$k$.  Using the standard construction
(e.g.~\cite{KKV01}), we transform $\mathsf{\mathsf{A}_\target}$ to an
equivalent B\"uchi automaton of size
$nk$.  Then we determinize the B\"uchi
automaton using e.g. the
Safra/Piterman-construction~\cite{Safra88,Piterman07} and obtain an
equivalent deterministic parity automaton with priorities $0$ to
$2nk-1$ and size
$\mathcal{O}(((nk)!)^2)$.  Finally we complement this parity automaton
by increasing every priority by 1, obtaining a deterministic parity
automaton
$\mathsf{B}_\target=(\detcarrier, \Sigma,\delta,v_0,\detprio)$ of size
$\mathcal{O}(((nk)!)^2)$, with priorities $1$ to
$2nk$ and with
\begin{align*}
L(\mathsf{B}_\target)=\overline{L(\mathsf{A}_\target)}=\overline{\mathsf{BadBranch}_\target}=:\mathsf{GoodBranch}_\target,
\end{align*}
i.e.
$\mathsf{B}_\target$ is a deterministic parity automaton that
accepts the words that encode paths that do not contain a least fixpoint whose satisfaction is deferred indefinitely. We use the labelling function $l:\detcarrier\to\Pow(\FLtarget)$ from the determinized
automaton.

\begin{remark}
It has been noted that the standard tracking automata for
\emph{alternation-free} formulas are, in fact, Co-B\"uchi automata~\cite{FriedmannLatteLange13,HausmannEA16} and that the tracking automata for \emph{aconjunctive} formulas are \emph{limit-deterministic} parity automata~\cite{HausmannEA18}. These considerably simpler automata can be determinized to
deterministic B\"uchi automata of size $3^n$ and to
deterministic parity automata of size $\mathcal{O}((nk)!)$ and
with $2nk$ priorities, respectively. This observation also holds
true for the tracking automata in this work so that for
formulas of suitable syntactic shape,
Lemma~\ref{lem:complexity} below yields accordingly lower bounds
on the runtime of our satisfiability checking algorithm.
\end{remark}

\section{Global Caching for the
Coalgebraic $\mu$-Calculus}\label{section:alg}

We now introduce a generic global caching algorithm that decides the
satisfiability problem of the coalgebraic $\mu$-calculus.  Given an
input formula $\target$, the algorithm expands the determinized and
complemented tracking automaton $\mathsf{B}_\target$ step by step and
propagates (un)satisfiability through this graph; the algorithm
terminates as soon as the initial node $v_0$ is marked as
(un)satisfiable.  The algorithm bears similarity to standard
game-based algorithms for
$\mu$-calculi~\cite{FriedmannLange13a,HausmannEA18,FontaineEA10};
however, it crucially deviates from these algorithms in the treatment
of modal steps: Intuitively, our algorithm decides whether it is
possible to remove some of the modal transitions as well as one of the
transitions from any reachable pair
$((\psi_1\vee\psi_2),0),((\psi_1\vee\psi_2),1)$ of disjunction
transitions within the automaton $\mathsf{B}_\target$ in such a way
that the resulting sub-automaton of $\mathsf{B}_\target$ that no
longer contains choices for disjunctions and (possibly) has a reduced
set of modal transitions is totally accepting, that is, accepts any
word for which there is an infinite run of the automaton.  In doing
so, it is crucial that the labels of state nodes $v$ in the reduced
automaton are \emph{one-step satisfied} in the set of states that are
reachable from $v$ by the remaining modal transitions.  This last
property is ensured by using instances of the so-called \emph{one-step
  satisfiability problem} to propagate (un)satisfiability over modal
transitions; these instances can often be solved in time
singly-exponential in $|\target|$, and in fact, this appears to be the
case for all currently known examples of decidable coalgebraic
$\mu$-calculi.  Previous work in~\cite{FontaineEA10} casts the modal
steps of satisfiability checking for coalgebraic $\mu$-calculi in
terms of satisfiability games but leads to a doubly-exponential number
of modal moves for one of the players and hence does not yield a
singly-exponential upper bound on satisfiability checking (unless a
suitable set of tableau rules is provided).

\begin{definition}[One-step satisfiability 
problem~\cite{KupkeEA15,FontaineEA10}]
  Let $V$ be a finite set, let $v\subseteq \Lambda(V)$ such that
  $a\neq b$ whenever $\hearts_1 a , \hearts_2 b\in v$, and let
  $U\subseteq\Pow(V)$.  The \emph{one-step satisfiability problem} for
  inputs $v$ and $U$ is to decide whether
  $TU\cap\sem{v}_1\neq\emptyset$, where
\begin{align*}
\sem{v}_1=\bigcap_{\hearts a\in v}\sem{\hearts}
\{u\in U\mid a\in u\}.
\end{align*}
We denote denote the time it takes to
solve the problem for input $v$ and $U$ by $t(|v|,|U|)$,
where $|v|=\sum_{\hearts a\in v}|\hearts T|$, the size of $v$, takes the representation of modal operators into account
and where $|U|$ is just the number of elements of $U$.
\end{definition}
  
\begin{remark}\label{rem:bounds}
  We keep the definition of the actual one-step logic as mentioned in
  the introduction somewhat implicit in the above definition of the
  one-step satisfiability problem. One can see that it contains two
  layers: a purely propositional layer embodied in~$U$, which
  postulates which propositional formulas over~$V$ are satisfiable,
  and a modal layer with nesting depth of modalities uniformly equal
  to~$1$, embodied in the set~$v$, which specifies constraints on an
  element of $TU$.
\end{remark}

\begin{example}\label{ex:onestepsat}
  For the one-step fragment of the standard modal $\mu$-calculus
  (Example~\ref{ex:logics}.1.), the one-step satisfiability problem
  for given $v\subseteq\Lambda(V)$ and $U\subseteq\Pow(V)$ consists
  in deciding whether there is a set $A\in \Pow (U)\cap\sem{v}_1$,
  that is, such that for each $\Diamond a\in v$, there is 
  $u\in A$ such that $\psi\in u$, and for each $\Box a\in v$ and each
  $u\in A$, $a\in u$.  Here we have $t(|v|,|U|)\leq |v|\cdot |U|$. 
  For the one-step fragment of the graded
  $\mu$-calculus (Example~\ref{ex:logics}.2.), the problem for input
  $v$ and $U$ consists in deciding whether there is a multiset
  $\theta\in\mathcal{B}(U)$ with $\theta\in\sem{v}_1$; the latter
  is the case if for each $\langle k\rangle a\in v$, we have
  $\sum_{u\in U_a}\theta(u)>k$, where
  $U_a=\{u\in U\mid a\in u\}$ and for each $[k]a\in v$,
  we have $\sum_{u\in U_{\neg a}}\theta(u)\leq k$, where
  $U_{\neg a}=\{u\in U\mid a\notin u\}$.
\end{example}
\noindent Let $l=2|\target|\cdot\mathsf{ad}(\target)$ denote the
number of priorities in $\mathsf{B}_\target$.  Nodes whose labels
consist exclusively of modal literals are referred to as
\emph{saturated nodes} or \emph{states}.  We denote the set of states
by $\mathsf{states}\subseteq \detcarrier$ and the set of pre-states,
that is, non-state nodes, by
$\mathsf{prestates}\subseteq \detcarrier$.  For each pre-state
$v\in\mathsf{prestates}$, we also fix a non-modal formula
$\psi_v\in l(v)$.  We now define $l$-ary set-valued functions $f$ and
$g$ that compute one-step (un)satisfiability w.r.t. their argument
sets.

\begin{definition}[One-step propagation] For sets $G\subseteq \detcarrier$
and $\mathbf{X}=X_1,\ldots, X_{l}
\subseteq G^l$, we put
\begin{align*}
f(\mathbf{X})=&\{v\in \mathsf{prestates}\mid \exists b\in \{0,1\}.\,
\delta(v,(\psi_v,b))\in X_{\detprio(v,(\psi_v,b))}\}\cup\\
&\{v\in \mathsf{states}\mid T(\bigcup_{1\leq i\leq l} X_i(v))\cap\sem{l(v)}_1\neq\emptyset\}\\
g(\mathbf{X})=&\{v\in \mathsf{prestates}\mid \forall b\in \{0,1\}.\,
\delta(v,(\psi_v,b))\notin X_{\detprio(v,(\psi_v,b))}\}\cup\\
&\{v\in \mathsf{states}\mid T(\bigcup_{1\leq i\leq l} X_i(v))\cap\sem{l(v)}_1=\emptyset\},
\end{align*}
where $\detprio(v,(\psi_v,b))$ abbreviates $\detprio(v,(\psi_v,b),\delta(v,(\psi_v,b)))$ and where
\begin{align*}
X_i(v)=\{l(u)\in X_i\mid \exists\sigma\in\mathsf{selections}.\,
\delta(v,\sigma)=\{u\}, \detprio(v,\sigma,u)=i\}.
\end{align*}
\end{definition}

\noindent

Since for states $v$, $l(v)\subseteq\Lambda(\mathbf{F})$
and since $\bigcup_{1\leq i\leq l} X_i(v)\subseteq\Pow(\FLtarget)$,
one-step propagation steps for states are just 
instances of the one-step satisfiability problem with $V=\FLtarget$.
Since states have at most $|\target|$ modal literals in their labels, these
instances can be solved in time $t(|\target|,2^{|\target|})$.

\begin{definition}[Propagation] Given a set $G$, we put
\begin{align*}
\mathbf{E}_G&=\eta_l X_l.\,\ldots \eta_2 X_2.\eta_1 X_1. f(\mathbf{X})\\
\mathbf{A}_G&=\overline{\eta_l} X_l\,\ldots \overline{\eta_2} X_2.\overline{\eta_1} X_1. g(\mathbf{X}),
\end{align*}
where $\mathbf{X}=X_1,\ldots,X_l$ for $X_i\subseteq G$,
where $\eta_i= \mu$ for odd $i$, $\eta_i=\nu$ for even $i$
and where $\overline{\eta}=\mu$ if $\eta=\nu$ and 
$\overline{\eta}=\nu$ if $\eta=\mu$.
\end{definition}

\noindent The set $\mathbf{E}_G$ contains nodes $v\in G$ for which
there are choices for all disjunction and modal transitions that are
reachable from $v$ within $G$ such that the labels of all reachable
states in the chosen sub-automaton of $\mathsf{B}_\target$ are
one-step satisfied and such that on all paths through the chosen
sub-automaton, the highest priority that is passed infinitely often is
even, the intuition being that no least fixpoint is
unfolded infinitely often without being dominated.  Dually, the set
$\mathbf{A}_G$ contains nodes for which there exist no such suitable
choices.

We recall that $v_0\in\detcarrier$ is the initial state of the
determinized and complemented tracking automaton $\mathsf{B}_\target$.
The algorithm expands $\mathsf{B}_\target$ step-by-step starting from
$v_0$; for pre-states $u$, the expansion step adds nodes according to
the fixed non-modal formula $\psi_u$ that is to be expanded next, and
for states, the expansion follows all (matching) selections.  The
order of expansion can be chosen freely, e.g. by heuristic
methods. Optional intermediate propagation steps can be used
judiciously to realize on-the-fly solving.

\begin{algorithm}[Global caching]\label{alg:global}
To decide the satisfiability of the input formula $\target$,
initialize the sets of
\emph{unexpanded} and \emph{expanded} nodes, $U=\{v_0\}$ and $G=\emptyset$,
respectively.
\begin{enumerate}
\item Expansion: Choose some unexpanded node $u\in U$, remove it from $U$ and add it to $G$.
If $u$ is a pre-state, then add the set $\{\delta(u,\sigma)\mid\sigma\in \Sigma\cap(\psi_u\times\{0,1\})\}$ to $U$.
If $u$ is a state, then add the set $\{\delta(u,\sigma)\mid
\sigma\in\mathsf{selections}\}$ to $U$.
\item Optional propagation: Compute $\mathbf{E}_G$ and/or $\mathbf{A}_G$. If
$v_0\in\mathbf{E}_G$, then return `satisfiable`, if $v_0\in\mathbf{A}_G$,
then return `unsatisfiable`.
\item If $U\neq\emptyset$, then continue with step 1.
\item Final propagation: Compute $\mathbf{E}_G$. If
$v_0\in\mathbf{E}_G$, then return `satisfiable`, otherwise return
`unsatisfiable`.
\end{enumerate}

\end{algorithm}

\begin{lemma}\label{lem:complexity}
Given a target formula $\target$ with $|\target|=n$ and
$\mathsf{ad}(\target)=k$, Algorithm~\ref{alg:global} terminates and 
runs in time $\mathcal{O}(((nk)!)^{4nk}\cdot t(n,2^n))$.

\end{lemma}

\begin{proof}
The loop of the algorithm expands the
determinized and complemented tracking automaton node by node and hence is executed at most 
$|\detcarrier|\in\mathcal{O}(((nk)!)^2)\in 2^{\mathcal{O}(nk\log n)}$ times. A single expansion step can be implemented
in time $\mathcal{O}(2^n)$ since propositional expansion is unproblematic
and for the modal expansion of a state $u$, all (matching) selections,
of which there are (at most) $2^n$, have to be considered. 
A single propagation step consists
in computing two fixpoints of nesting depth $l=2nk$ 
of the functions $f$ and $g$ over $\Pow(\detcarrier)$ and
can hence be implemented in time $2(|\detcarrier|^{2nk}\cdot t(n,2^n))\in\mathcal{O}(((nk!)^2)^{2nk}\cdot t(n,2^n))
\in 
2^{\mathcal{O}(n^2k^2\log n+\log(t(n,2^n)))}$,
noting that a single computation of
$f(\mathbf{X})$ and $g(\mathbf{X})$ for a tuple $\mathbf{X}\subseteq (\detcarrier)^k$ 
can be implemented in time $\mathcal{O}(t(n,2^n))$:
for pre-states, the one-step propagation is unproblematic 
and for states, it consists in solving the one-step satisfiability problem
with inputs of size at most $n$ and $2^n$,
as explained above.
Thus the complexity of the whole
algorithm is dominated by the complexity of the propagation step.\qed
\end{proof}

\begin{corollary}
If the one-step satisfiability problem of a coalgebraic logic
for inputs $v$ and $U$ with $|U|\leq 2^{|v|}$ can be solved in
time $t(|v|,|U|)\leq 2^{p(|v|)}\cdot p'(|U|)$,
where $p$ and $p'$ are some polynomial functions, 
then the satisfiability problem of the 
$\mu$-calculus over the logic is in \ExpTime.
\end{corollary}

\noindent The complexity bounds obtained by our current semantic
approach thus subsume the earlier bounds obtained by the tableau-based
approaches in~\cite{CirsteaEA11a,HausmannEA16,HausmannEA18}
but also cover new example logics. In
particular we have
\begin{lemma}
The satisfiability problems of the following logics are in \ExpTime:
\begin{enumerate}
\item the standard $\mu$-calculus,
\item the graded $\mu$-calculus, 
\item the (two-valued) probabilistic $\mu$-calculus,
\item the Presburger $\mu$-calculus, 
\item the (two-valued) probabilistic $\mu$-calculus extended with polynomial inequalities
\end{enumerate}
\end{lemma}
\begin{proof}
  It suffices to show that the respective one-step satisfiability
  problems can be solved in time $t(n,2^n)\leq 2^{p(n)}\cdot p'(2^n)$,
  that is, in time singly exponential in $n$,
  for inputs $v$ and $U$ of sizes $|v|\leq n$ and $|U|\leq 2^n$.
  While this follows by relatively easy
  arguments (using known bounds on sizes of solutions of systems of
  real or integer linear inequalities) for all of our examples, we
  import most of the results from previous work for brevity.  For
  standard Kripke logic, we have $p(x)=\log x$ and $p'(x)=x$,
  see Example~\ref{ex:onestepsat}. 
  For the one-step satisfiability
  problem of graded modal logic, by Lemma
  1 in~\cite{KupfermanEA02}, we have $t(n,2^n)\leq (2n+2)^n\leq
  2^{n\log (2n+2)} $ and choose, e.g.,
  $p(x)=x\log(2x+2)$ and $p'(x)=x$.  The corresponding properties for
  (two-valued) probabilistic modal logic and the two arithmetic logics
  (items 4. and 5.) are shown in Example~7
  in~\cite{KupkeEA15}.  \qed
\end{proof}

\begin{remark}
We also obtain a polynomial bound on branching width
in models for all our example logics 
simply by importing Lemma 6 and the
observations in Example 7 from~\cite{KupkeEA15}. 
With exception of the standard $\mu$-calculus,
this bound appears to be novel for all example logics in this work.
\end{remark}

\section{Soundness and Completeness}\label{sec:correctness}

We now prove the central result, that is, the total correctness of Algorithm~\ref{alg:global}. As the sets $\mathbf{E}_G$ and $\mathbf{A}_G$ grow monotonically
with $G$, it suffices to prove equivalence of satisfiability and
containment of the initial node $v_0$ in $\mathbf{E}:=\mathbf{E}_{\detcarrier}$.

\begin{theorem}[Soundness and completeness]
We have 
\begin{align*}
v_0\in\mathbf{E}\text{ if and only if $\target$ is satisfiable}.
\end{align*}
\end{theorem}
\begin{proof}
By Corollary~\ref{cor:wintimeouts}, it suffices to show that
there is a \emph{pre-semi-tableau} (see Definition~\ref{defn:presemtab}) for $\target$ with \emph{unfolding timeouts} (see Definition~\ref{defn:unfto})
if and only if $\target$ is satisfiable.
So let there be a pre-semi-tableau for $\target$ with unfolding timeouts.
We use the Existence 
Lemma (Lemma~\ref{lem:existence}) to obtain a \emph{strongly coherent coalgebra} (see Definition~\ref{defn:strongcoherence}) which by the Truth Lemma (Lemma~\ref{lem:truth}) is a model for $\target$.
For the converse direction, let $\target$ be satisfiable.
We use
Lemma~\ref{lem:modeltotab}
to extract a pre-semi-tableau with unfolding timeouts from the model.
\qed
\end{proof}

\begin{definition}[Pre-semi-tableau]\label{defn:presemtab}
Given two sets $A$ and $B$, a ternary relation $R\subseteq A\times B\times A$ and two elements $a\in A$, $b\in B$,
we put $R(a)=\{a'\in A\mid \exists b\in B.\, (a,b,a')\in R\}$
and
$R(a,b)=\{a'\in A\mid (a,b,a')\in R\}$.
Let $W\subseteq \detcarrier$ 
be a set of nodes labelled with formulas from $\FLtarget$
and put $U=W\cap\mathsf{prestates}$ and $V=W\cap\mathsf{states}$.
Given a ternary relation $L\subseteq W\times \Sigma\times W$,
the pair $(W,L)$ is a \emph{pre-semi-tableau} for $\target$ if $L\subseteq \delta$
and
 for all $v\in V$, we have $T(L(v))\cap\sem{l(v)}_1\neq\emptyset$,
for all $u\in U$, there is exactly one $b\in\{0,1\}$ such that 
$L(u,(\psi_u,b))=\delta(u,(\psi_u,b))$ and for all other $\sigma\in\Sigma$, $L(u,\sigma)=\emptyset$ and
there is no $L$-cycle that contains only elements from $U$.
A \emph{path} through a pre-semi-tableau is an infinite sequence
$(v_0,\sigma_0),(v_1,\sigma_1),\ldots\in (W\times \Sigma)^\omega$ such that
for all $i$, $v_{i+1}\in L(v_i,\sigma_i)$.
We denote \emph{the} first state
that is reachable by zero or more $L$-steps from a node $v\in W$ by $\lceil v\rceil$ (since
there is no $L$-cycle that contains only elements from $U$,
such a state always exists).
\end{definition}

\noindent Given a state $v$, the relation $L$ of a pre-semi-tableau
thus picks a set $L(v)$ of nodes over which a coherent observation for
$v$ can be built; given a pre-state $u$, $L$ picks a single (pre)state
that is obtained from $u$ by transforming the formula $\psi_u$.

\begin{definition}[Tracking timeouts]\label{defn:trackto}
Given a path $\rho=(v_0,\sigma_0),(v_1,\sigma_1),\ldots$ through a pre-semi-tableau,
we say that priority $i$ \emph{occurs} (at position $j$) in 
$\rho$ if $\beta(v_j, \sigma_j,v_{j+1})=i$, recalling that
$\beta$ is the priority function of the determinised and
complemented tracking automaton $\mathsf{B}_\target$.
Then the path $\rho$ has \emph{tracking timeouts} $\overline{m}=(m_l,\ldots,m_1)$, if for each
odd $1\leq i< l$, priority $i$ occurs at most $m_i$ times in $\rho$ before
some priority greater than $i$ occurs in $\rho$.
Nothing is said about $m_i$ for even $i$, which are in fact irrelevant
and serve only to ease notation.
An element $w\in W$ has \emph{tracking timeouts} $\overline{m}$ in some pre-semi-tableau $(W,L)$
if every path through $(W,L)$ that starts at $w$ has tracking timeouts $\overline{m}$.
A pre-semi-tableau $(W,L)$ has \emph{tracking timeouts} if there is, for each $w\in W$,
some vector $\overline{m}$ such that $w$ has tracking timeouts $\overline{m}$.

\end{definition}

\noindent Intuitively, a pre-semi-tableau $(W,L)$ has tracking
timeouts if every word that encodes an infinite $L$-path through $W$
is accepted by $\mathsf{B}_\target$. We recall that a run of
$\mathsf{B}_\target$ is accepting if the encoded path does not contain
a trace that unfolds some least fixpoint formula infinitely often
without having it dominated.

\begin{definition}[Unfolding timeouts]\label{defn:unfto}
Given a path $\rho=(v_0,\sigma_0),(v_1,\sigma_1),\ldots$ through a pre-semi-tableau
and a sequence of formulas $\Psi=\psi_0,\psi_1,\ldots$, we say that $\Psi$ is a \emph{trace}
of $\psi_0$ in $\rho$ (we also say that $\rho$ \emph{contains} $\Psi$) if $\psi_0\in l(v_0)$
and for all $i>0$, $\psi_i\in l(v_i)\cap\Delta(\psi_{i-1},\sigma_{i-1})$.
For $i$ with $\psi_i=\eta X.\psi$ for some fixpoint variable $X$ and some formula $\psi$,
we say that $\Psi$ \emph{unfolds at level} $\mathsf{ad}(\psi_i)$ at position~$i$.
Then the trace $\Psi$ has \emph{unfolding timeouts} $\overline{m}=(m_k,\ldots,m_1)$ for $\psi_0$ if for each
odd $1\leq i\leq k$, $\Psi$ unfolds at most $m_i$ times at level $i$ before
$\Psi$ unfolds at some level greater than $i$. 
Again the unfolding timeouts for even $i$, that is, for greatest fixpoints,
are irrelevant.
The path $\rho$ has \emph{unfolding timeouts} for $\psi_0$
if there is, for all its traces $\Psi$ of $\psi_0$, some vector $\overline{m}$ such that
$\Psi$ has unfolding timeouts $\overline{m}$ for $\psi_0$.
Given a pre-semi-tableau $(W,L)$,
a node $w\in W$ has \emph{unfolding timeouts} $\overline{m}$ for some formula
$\psi$ 
if every path through $(W,L)$ that starts at $w$ and contains a trace
of $\psi$ has unfolding timeouts $\overline{m}$ for $\psi$.
A pre-semi-tableau $(W,L)$ has \emph{unfolding timeouts} if
for each element $w\in W$ and each formula $\psi\in l(v)$, 
there is
some vector $\overline{m}$ such that $w$ has unfoldings timeouts $\overline{m}$
for $\psi$. We denote 
the set of states that have unfolding timeouts $\overline{m}$
for $\psi$ by $\mathsf{uto}(\psi,\overline{m})\subseteq W$.

\end{definition}

\noindent A pre-semi-tableau $(W,L)$ has unfolding timeouts if for all
words that encode an infinite $L$-path through $W$, all runs of the
nondeterministic tracking automaton $\mathsf{A}_\target$ on the word
are \emph{non}-accepting. We recall that a run of $\mathsf{A}_\target$
is accepting if it unfolds some least fixpoint infinitely often
without having it dominated.

\begin{lemma}\label{lem:timeouts}
Let $(W,L)$ be a pre-semi-tableau. Then $(W,L)$ has tracking timeouts if and only if it has
unfolding timeouts.
\end{lemma}
\begin{proof}
We recall that $\mathsf{B}_\target$ is obtained from
$\mathsf{A}_\target$ by determinization and subsequent
complementation so that we have $L(\mathsf{B}_\target)=\overline{L(\mathsf{A}_\target)}$.
The result thus follows directly from the fact that
having tracking timeouts ensures that $\mathsf{B}_\target$ accepts all
words that encode a path in $(W,L)$
while having unfolding timeouts ensures that
$\mathsf{A}_\target$ does not accept any word that encodes
a path in $(W,L)$.
\qed
\end{proof}

\begin{lemma}\label{lem:winner}
We have $v_0\in\mathbf{E}$ if and only if there is a pre-semi-tableau for $\chi$
that has tracking timeouts.
\end{lemma}
\begin{proof}[Sketch] If $v_0\in\mathbf{E}$, then the definition of the function $f$
ensures the existence of suitable transitions in $\mathsf{B}_\target$
that can be used to define a pre-semi-tableau $(W,L)$. 
The definition of $(W,L)$ has to be executed 
in a nested inductive-coinductive way, relying on the fact that
$v_0$ is contained in the nested fixpoint $\mathbf{E}$ to ensure
that $(W,L)$ has tracking timeouts. For the converse direction, let
there be a pre-semi-tableau for $\chi$ that has tracking timeouts.
Then $v_0$ has tracking timeouts $\overline{m}$ for some $\overline{m}$.
We show $v_0\in \mathbf{E}$ by nested induction and coinduction, using
$\overline{m}$ as termination measure for the induction parts.
The full proof can be found in the appendix.
 \qed
\end{proof}

\noindent Combining Lemmas~\ref{lem:winner} and~\ref{lem:timeouts}, we
obtain
\begin{corollary}\label{cor:wintimeouts}
We have $v_0\in\mathbf{E}$ if and only if there is a pre-semi-tableau for $\chi$
that has unfolding timeouts.
\end{corollary}

\begin{definition}
Given a pre-semi-tableau $(W,L)$ with set of states $V$, we put
\begin{align*}
\psem{\psi}&=\{v\in V\mid l(v)\vdash_{\mathsf{PL}}\psi\} &
\psem{\psi}_{\overline{m}}&=\psem{\psi}\cap\{\lceil u\rceil \in V\mid u\in\mathsf{uto}(\psi,\overline{m})\}
\end{align*}
where $\psi\in\FLtarget$, where $\vdash_\mathsf{PL}$ denotes
propositional entailment and where $\overline{m}$ is a vector of $k$ natural numbers.
\end{definition}
Thus we have $v\in\sem{\psi}_{\overline{m}}$ if there is a node
$u\in W$ 
such that $\lceil u\rceil=v$ and
$u$ has timeouts $\overline{m}$ for $\psi$. This serves to ease 
the proofs of the upcoming Existence and Truth Lemmas as it anchors the 
timeout vector $\overline{m}$ at the node $u$ instead of anchoring it at the state $v$ 
which may not have timeouts $\overline{m}$ for $\psi$ 
(namely, if a greatest fixpoint is unfolded on
the $L$-path from $u$ to $v$).

\begin{definition}[Strong coherence]\label{defn:strongcoherence}
Let $(W,L)$ be a pre-semi-tableau with unfolding timeouts and set of states $V$.
A coalgebra $\mathcal{C}=(V,\xi)$ 
is \emph{strongly coherent} if for all states $v\in V$, for all formulas
$\hearts\psi\in \FLtarget$ and for all timeout-vectors $\overline{m}$,
\begin{align*}
v\in\psem{\hearts\psi}_{\overline{m}} \text{ implies }
\xi(v)\in\sem{\hearts}(\psem{\psi}_{\overline{m}}).
\end{align*}
\end{definition}

\begin{lemma}[Existence]\label{lem:existence}
Let $(W,L)$ be a pre-semi-tableau with set of states $V$
that has unfolding timeouts.
Then there is a strongly coherent coalgebra over $V$.
\end{lemma}
\begin{proof}
Let $v\in V$ be a state with $l(v)=\{\hearts_1\psi_1,\ldots,\hearts_n\psi_n\}$
and put $L_v=\{(\sigma_i,v_i)\in \Sigma\times W\mid (v,\sigma_i,v_i)\in L\}$.
As $(W,L)$ is a pre-semi-tableau, we have $T(L(v))\cap\sem{l(v)}_1\neq \emptyset$,
so we put $\xi(v)=t$ for some $t\in T(L(v))\cap\sem{l(v)}_1\neq \emptyset$.
Let $i$ be a number for which there is a vector $\overline{m}'$
such that $v\in\mathsf{uto}(\hearts_i\psi_i,\overline{m}')$. Such a vector
exists since $(W,L)$ has unfolding timeouts. 
Now let $\overline{m}$ denote the \emph{least} such vector (by lexicographic ordering). It suffices to show
that $t\in\sem{\hearts_i}(\psem{\psi_i}_{\overline{m}})$.
By construction, we have $t\in\sem{\hearts_i}(\{u\in W\mid \psi_i\in l(v)\}\cap L(v))$
and hence $t\in\sem{\hearts_i}(F)$ where
\begin{align*}
F=\psem{\psi_i}\cap\{\lceil w_j\rceil \in V\mid (w_j,\sigma_j)\in L_v, \psi_i\in \Delta(\hearts_i\psi_i,\sigma_j)\}.
\end{align*}
As $v\in \mathsf{uto}(\hearts_i\psi_i,\overline{m})$, every infinite $L$-path
$(v,\sigma_j),(w_j,\sigma'),\ldots$ such that $(\sigma_j,w_j)\in L(v)$ and that
contains a trace of $\hearts_i\psi_i$
has unfolding timeouts $\overline{m}$
for $\hearts_i\psi_i$ and $\psi_i\in\Delta(\hearts_i\psi_i,\sigma_j)$. Hence all such $w_j\in L(v)$ have unfolding timeouts $\overline{m}$
for $\psi_i$. Thus
$F\subseteq\psem{\psi_i}\cap\{\lceil w\rceil \in V\mid w\in\mathsf{uto}(\psi_i,\overline{m})\}=\psem{\psi_i}_{\overline{m}}$, as required.
\qed
\end{proof}

 \begin{definition}[Timed-out satisfaction]
Given sets $W$, $U\subseteq W$, a function $f:\Pow(W)\to\Pow(W)$ and an
ordinal number $\lambda$, we define
\begin{align*}
f^\lambda(U)=\begin{cases}
U & \text{ if } \lambda=0\\
f(f^{\lambda'}(U)) & \text{ if } \lambda=\lambda'+1\\
\bigcup_{k<\lambda} f^k(U) & \text{ if $\lambda$ is a limit-ordinal}
\end{cases}
\end{align*}
Given a least fixpoint formula $\mu X.\psi$ with $\mu X.\sem{\psi}_i^X=
(\sem{\psi}_i^X)^\lambda(\emptyset)$ where $i:\mathbf{V}\to\Pow(C)$ valuates
fixpoint variables, as usual, we say that $x\in(\sem{\psi}_i^X)^\lambda(\emptyset)$
\emph{satisfies $\mu X.\psi$ with timeout} $\lambda$ (under $i$)
and write $x\in\sem{\mu X.\psi}^{\lambda}_i$. 
For all models $(C,\xi)$ and ordinal numbers $\lambda$ with
$|C|=\lambda$, all states $x\in C$ and all least fixpoint formulas
$\mu X.\psi$ such that $x\models \mu X.\psi$, we have $x\in\sem{\mu X.\psi}^\lambda_\epsilon$ where $\epsilon$ denotes the empty valuation of
fixpoint variables. Formulas $\psi$
from nested fixpoints 
are satisfied (under $i$) with vectors of ordinal numbers
$\overline{\lambda}=(\lambda_k,\ldots,\lambda_j)$, ordered
from outermost fixpoint to innermost fixpoint, as timeouts;
here, $i$ is assumed to valuate $k-j$ alternating fixpoint variables.
Again the timeouts for greatest fixpoint variables are irrelevant
and serve only to ease notation.
We write $x\in\sem{\psi}_i^{\overline{\lambda}}$ to indicate that 
$x$ satisfies $\psi$ with nested timeouts
$\overline{\lambda}$ (under $i$).
\end{definition}
We have $\sem{\mu X.\psi}_i^{0}=\emptyset$,
$\sem{\mu X.\psi}_i^{\lambda+1}=\sem{\psi[X\mapsto \mu X.\psi]}_i^\lambda$
for ordinals $\lambda$ and $\sem{\mu X.\psi}_i^{\lambda}=\bigcup_{k<\lambda}
\sem{\mu X.\psi}_i^k$ for limit-ordinals $\lambda$.

In strongly coherent coalgebras, all least fixpoint literals are satisfied after
finitely many unfolding steps:

\begin{lemma}[Truth]\label{lem:truth}In strongly coherent coalgebras, we have
that for all $\psi\in\FLtarget$, 
\begin{align*}
\psem{\psi}\subseteq\sem{\psi}.
\end{align*}
\end{lemma}

\begin{proof}[Sketch] This proof is standard for $\mu$-calculi; we use nested induction and coinduction to show eventual satisfaction of all least fixpoint formulas that occur in some label,
using unfolding timeouts as termination measure for the induction
parts of the proof. The full proof can be found in the appendix. \qed
\end{proof}

\begin{lemma}[Soundness]\label{lem:modeltotab} Let $\target$ be satisfiable. Then
a pre-semi-tableau for $\target$ with unfolding timeouts can be constructed
over a subset of $\detcarrier$.
\end{lemma}
\begin{proof}[Sketch]
Relying on the information from
some fixed model for $\target$,
we construct a pre-semi-tableau for $\target$ by choosing
usable transitions in the determinized and complemented tracking automaton
$\mathsf{B}_\target$. The transitions
have to be chosen in a way that preserves the timeouts with which
formulas are satisfied in the model. Then we use nested transfinite induction
and coinduction to show that the constructed pre-semi-tableau has unfolding timeouts, using the timeout vectors $(\lambda_k,\ldots,\lambda_j)$ 
with which formulas are satisfied in the model
as termination measure for the transfinite induction parts of the proof.
The full proof can be found in the appendix. \qed
\end{proof}

\begin{corollary}
Let $\target$ be a satisfiable coalgebraic $\mu$-calculus formula.
Then $\target$ has a model of size $\mathcal{O}(((nk)!)^2)\in 2^{\mathcal{O}(nk\log n)}$.
\end{corollary}
\begin{proof}
Let $\target$ be satisfiable. By soundness, we have $v_0\in\mathbf{E}$.
The model for $\target$ that is constructed during the completeness proof
is built over $\mathbf{E}\cap\mathsf{states}\subseteq \detcarrier$.
The stated bound follows since $|\detcarrier|\in\mathcal{O}(((nk)!)^2)$.
\qed
\end{proof}




\section{Conclusion}\label{section:conclusion}

We have shown that the satisfiability problem of the coalgebraic
$\mu$-calculus is in \ExpTime if the corresponding one-step
satisfiability problem can be solved in time singly exponential in $n$
for inputs $v,U$ of sizes $n$ and $2^n$. Prominent examples 
where this is the case include
the graded $\mu$-calculus, the (two-valued) probabilistic
$\mu$-calculus, the Presburger $\mu$-calculus, and the extension of
the two-valued probabilistic $\mu$-calculus with polynomial
inequalities; the \ExpTime bound appears to be novel
for the last two logics.  We also have presented a generic
satisfiability algorithm that realizes the singly exponential time upper
bound under the stated assumption and is suitable for practical use since it
supports global caching and on-the-fly solving. Moreover, we obtained
a novel singly-exponential bound on minimum model size of
satisfiable formulas for \emph{all}
decidable coalgebraic $\mu$-calculi and a polynomial bound on the
branching width in models for all example logics mentioned above.


\newpage

\bibliography{coalgml}

\newpage
\appendix

\section{Appendix: Omitted Proofs and Lemmas}

\textit{Full proof of Lemma~\ref{lem:winner}:} Let $v_0\in\mathbf{E}$. 
All elements $v$ of the nested fixpoint $\mathbf{E}$ (which has nesting depth $l$) have nested timeouts 
$\overline{m}=(m_l,\ldots,m_1)$ with $m_i\leq|\detcarrier|$ that ensure
that $v$ can be shown to be contained in $\mathbf{E}$ while, for all odd
$1\leq i\leq l$, unfolding fixpoints $\eta_i X_i.\eta_{i-1} X_{i-1}.\ldots
\eta_1 X_1. f(X_1,\ldots,X_i,M(m_l,\ldots,m_{i+1}),\ldots,M(m_l))$
at most $m_i$ times before unfolding the fixpoint for some number
greater than $i$; here we use $M(\overline{m})$ to denote the elements 
of $\mathbf{E}$ that have nested timeouts $\overline{m'}$ where
$\overline{m}'$
is just $\overline{m}$ if the length of $\overline{m}$ is even and where
$\overline{m}'$
is obtained from $\overline{m}$ by decreasing the last element
of $\overline{m}$ by one if the length of $\overline{m}$ is odd.
Thus there is, for each $v\in\mathbf{E}$, some least (by lexicographic ordering) vector, denoted by $\overline{m}_v$, such that $v$ has 
nested timeouts 
$\overline{m}_v$. Now we define a pre-semi-tableau for $\target$ over
$\mathbf{E}$: Put $V=\mathbf{E}\cap\mathsf{states}$,
$U=\mathbf{E}\cap\mathsf{prestates}$ and $W=V\cup U$. 
Given a pre-state $u\in U$ with nested timeouts $\overline{m}_u$,
since $u\in \mathbf{E}$, there is a $b$ such that $\delta(u,(\psi_u,b))
\in M(m_l,\ldots,m_{\beta(u,(\psi_u,b),\delta(u,(\psi_u,b)))})$. We put $L(u,(\psi_u,b))=
\delta(u,(\psi_u,b))$ and $L(u,\sigma)=\emptyset$ for all $\sigma\in\Sigma$
with $\sigma\neq(\psi_u,b)$. Given a state $v\in V$ with
nested timeouts $\overline{m}_v$, since $v\in\mathbf{E}$,
we have $T(\bigcup_{1\leq i<l} M_i(v))\cap\sem{l(v)}_1\neq\emptyset$
where $M_i(v)$ is the set of the labels of nodes
$u\in M(m_l,\ldots,m_i)$ for which there is a selection
$\sigma$ with $\beta(v,\sigma,\delta(v,\sigma))=i$.
For each such $\sigma$, we put $L(v,\sigma)=\delta(v,\sigma)$; 
for all other $\sigma\in \Sigma$, we put $L(v,\sigma)=\emptyset$.
As $L\subseteq \delta$, as $T(L(v))\cap\sem{l(v)}_1\neq\emptyset$ for
states $v\in V$, as for pre-states
$u\in U$, there is exactly one $b\in\{0,1\}$ such that
$L(u,(\psi_u,b))=\{v\}$ and for all other $\sigma\in\Sigma$,
$L(u,\sigma)=\emptyset$, and as there is -- by guardedness of fixpoint-variables --
no $L$-cycle in $U$, $(W,L)$ is a pre-semi-tableau. Since we
constructed $L$ in such a way that nested timeouts are respected
and since nested and tracking timeouts both are defined by means of $\beta$,
$(W,L)$ has tracking timeouts too.

The proof for the converse direction is analogous: 
Let $(W,L)$ be a pre-semi-tableau with tracking timeouts, and
with set $V$ of states and set $U$ of
pre-states. To show that $W$ is contained in the fixpoint 
$\mathbf{E}$, we proceed by nested induction and coinduction
using tracking timeouts as termination measure for the induction parts of
the proof. For pre-states $u\in U$ with minimal
tracking timeouts $\overline{m}_u=(m_l,\ldots,m_j)$, there is a single $b$
such that $L(u,(\psi_u,b))=\delta(u,\psi_u,b)=:v$ where 
$v$ has tracking timeouts $\overline{m}_v$ that are obtained
from $\overline{m}_u$ by decreasing the $i$-th element by 1 if
$\beta(u,(\psi,b),\delta(u,(\psi,b)))=i$ for $i$ odd; thus
$u$ is contained in $f(M(m_l,\ldots,m_1),\ldots,M(m_l))$, as required.
For states $v\in V$ with minimal
tracking timeouts $\overline{m}_v=(m_l,\ldots,m_j)$, we have $T(L(v))\cap\sem{l(v)}_1\neq\emptyset$ where, again, the timeouts $\overline{m}_u$
for nodes $u\in L(v)$ are determined by $\overline{m}_v$ and $\beta(v,\sigma,\delta(v,\sigma))$, where $\sigma$ is such that $\delta(v,\sigma)=u$.
Again we have $v\in f(M(m_l,\ldots,m_1),\ldots,M(m,l))$, as required.
\qed\medskip\medskip

\noindent \textit{Full proof of Lemma~\ref{lem:truth}:}
Let $(V,\xi)$ be a strongly coherent coalgebra that is built over a pre-semi tableau
$(W,L)$ with unfolding timeouts and set of states $V$ and let $v\in\psem{\psi}$.
We proceed by induction over $\psi$.  The interesting cases are the cases with $\psi=\eta X.\psi'$ 
for some fixpoint variable $X\in\mathbf{V}$ and formula $\psi'$. In this case we start
a second induction over the number $o$ of closed fixpoint operators that are a subformula of $\psi$.
If $o=1$, then $\psi'$ contains no further closed fixpoint operators. If $o>1$, then we have, for
any fixpoint formula $\eta' Y.\phi$ that is a subformula of $\psi'$, that
$\psem{\eta' Y.\phi}\subseteq \sem{\eta' Y.\phi}$ by the induction
hypothesis. In both cases we are done if we reach a closed formula in the proof by nested induction
and coinduction below.
Since $(W,L)$ has unfolding timeouts, there is some vector $\overline{m}'$ such that
$v\in\psem{\psi}_{\overline{m}'}$.  
It thus suffices to show that for all timeout vectors $\overline{m}$ and formulas
$\psi\in\FLtarget$, we have $\psem{\psi}_{\overline{m}}\subseteq \sem{\psi}_{\overline{m}}$. 
We show this by nested induction and coinduction, using $(\overline{m},|\psi|)$ as
termination measure and distinguishing upon the shape of $\psi$. We consider just
the two interesting cases where $\psi=\hearts\psi'$ and where $\psi=\mu X.\psi'$. The former case is directly
finished by strong coherence of $(V,\xi)$. In the latter case we use the fact that
the unfolding of least fixpoint formulas reduces unfolding timeouts so that we have
$v\in\psem{\psi'[X\mapsto \psi]}_{\overline{m}'}$ where $\overline{m}'$ is obtained
from $\overline{m}$ by reducing $m_{\mathsf{ad}(\psi)}$ by 1 and leaving all other
timeouts unchanged so that we have $\overline{m'}<\overline{m}$
and the induction hypothesis finishes the case. Thus we have
shown that for all timeout vectors $\overline{m}$ and formulas
$\psi\in\FLtarget$, $\psem{\psi}_{\overline{m}}\subseteq \sem{\psi}_{\overline{m}}$ which in particular implies
 $v\in\sem{\psi}_{\overline{m}'}$, as required. 
\qed\medskip\medskip

\begin{lemma}\label{lemm:transf}
Let $G$ be a finite set, let $f:\Pow(G)\to\Pow(G)$ be a monotone
function and let $n$ be a number such that $\mu f=f^n(\emptyset)$.
Then we have that for all ordinal numbers $\lambda\geq n$,
\begin{align*}
\mu f = f^\lambda(\emptyset).
\end{align*}
\end{lemma}
\begin{proof}
The proof is by transfinite induction over $\lambda$. 
If $\lambda=0$, then $n=0$ and
$\mu f=f^0(\emptyset)=\emptyset$ so that we are done.
If $\lambda=\lambda'+1$, then we have $f^{\lambda'+1}(\emptyset)=f(f^{\lambda'}(\emptyset))=
f(\mu f)=\mu f$, where the second equality is by the induction 
hypothesis and the third equality holds since $\mu f$ is a fixpoint.
If $\lambda$ is a limit-ordinal, then we have $f^\lambda(\emptyset)=
\bigcup_{k< \lambda}f^k(\emptyset)$. By the induction hypothesis,
we have $f^k(\emptyset)=\mu f$ for all $k<\lambda$ so that we are done.\qed
\end{proof}

\noindent \textit{Full proof of Lemma~\ref{lem:modeltotab}:}
Let $(C,\xi)$ be a coalgebra and let $x\in C$ be a state with $x\models\target$.
We put $M=\{v\in \detcarrier\mid \exists y\in C.\, y\models l(v)\}$,
$V=M\cap\mathsf{states}$, $U=M\cap\mathsf{prestates}$
and define a pre-semi-tableau over $M$ in a timeout-respecting manner:
Let $u\in U$ and $y\in C$ with $y\models u$.
Also let $\overline{m}$ be the least (by lexicographic ordering) 
vector of ordinal numbers such that
$y\in\sem{\psi_u}_i^{\overline{m}}$.
If $\psi_u$ is a disjunction $\psi_1\vee\psi_2$, then we choose $b$ such that
$y\in\sem{\psi_b}_i^{\overline{m}}$. Otherwise, we put $b=0$.
Then we put $L(u,(\psi_u,b))=\delta(u,(\psi_u,b))$
and $L(u,\sigma)=\emptyset$ for all other $\sigma\in\Sigma$.
For $v\in V$ and $y\in C$ with $y\models v$,
we have $\xi(x)\in\bigcap_{\hearts\psi\in l(v)}\sem{\hearts}\sem{\psi}$.
For each $\hearts\psi\in l(v)$, there even is a least vector $\overline{m}_{\hearts\psi}$ of ordinal numbers such that
$\xi(x)\in\sem{\hearts}(\sem{\psi}_i^{\overline{m}_{\hearts\psi}})$.
For each $\sigma\in\mathsf{selections}$, we define $L(v,\sigma)=\delta(v,\sigma)$. For all other $\sigma\in\Sigma$, we put $L(v,\sigma)=\emptyset$.
Then we have $L\subseteq \delta$, $T(L(v))\cap\sem{l(v)}_1\neq\emptyset$ for
states $v\in V$, for pre-states
$u\in U$, there is exactly one $b\in\{0,1\}$ such that
$L(u,(\psi_u,b))=\{v\}$ and for all other $\sigma\in\Sigma$,
$L(u,\sigma)=\emptyset$, and there is -- by guardedness of fixpoint-variables --
no $L$-cycle in $U$; thus $(M,L)$ is a pre-semi-tableau.

It remains to show that $(M,L)$ has unfolding timeouts. We let $v\in M$,
$\psi\in l(v)$ and $y\in C$ with $y\models l(v)$ where $y\in\sem{\psi}_i^{\overline{m}}$ and proceed by nested transfinite induction and coinduction
over $\overline{m}$. We have to show that if a least fixpoint formula that
is a subformula of $\psi$
is satisfied at $y$ in the model with unfolding timeout $\lambda$, then this
fixpoint is unfolded in $(M,L)$ at most finitely often 
before being satisfied when starting from $v$.
The latter can be shown by proving containment of $v$ in the least fixpoint
of a suitable function $h$, i.e. in $\mu h=h^n(\emptyset)$ for
some finite number $n$. By Lemma~\ref{lemm:transf}, we have 
$h^n(\emptyset)=h^\lambda(\emptyset)$. Thus we proceed by transfinite induction over $\overline{m}=(m_k,\ldots,m_1)$ and distinguish upon the shape of $\psi$. 
The interesting case is the case with $\psi=\mu X.\psi'$ for some
fixpoint variable $X\in\mathbf{V}$ and formula $\psi'$ with $j:=\mathsf{ad}(\psi')$ even. If
$m_{j-1}$ is $0$, then we have $\sem{\mu X.\psi'}_i^0=\emptyset$ so that
there is nothing to show. If $m_{j-1}$ is $\lambda'+1$ for some ordinal
number $\lambda'$, 
then the induction hypothesis finishes the proof, using the fact 
that $\sem{\mu X.\psi}_i^{\lambda'+1}=\sem{\psi[X\mapsto \mu X.\psi]}_i^{\lambda'}$
and that unfolding timeouts in the pre-semi-tableau $(M,L)$ count unfoldings in the overlaying model. If $m_{j-1}$ is a limit ordinal, then we have 
$y\in \sem{\mu X.\psi}^{m_{j-1}}(\emptyset)=
\bigcup_{k<m_{j-1}}\sem{\mu X.\psi}^k(\emptyset)$. 
By the induction hypothesis, $y\in \sem{\mu X.\psi}^{k}(\emptyset)$ implies $v\in h^k(\emptyset)$ for all $k<m_{j-1}$
so that we have $v\in\bigcup_{k<m_{j-1}}h^k(\emptyset)=h^{m_{j-1}}(\emptyset)$, as required.\qed





\end{document}